\documentclass[a4paper]{article}
\usepackage{a4wide}
\usepackage{amsthm}
\theoremstyle{plain}
\newtheorem{theorem}{Theorem}[section]

\newtheorem{proposition}[theorem]{Proposition}
\newtheorem{lemma}[theorem]{Lemma}
\newtheorem{corollary}[theorem]{Corollary}

\theoremstyle{definition}
\newtheorem{definition}[theorem]{Definition}
\newtheorem{example}[theorem]{Example}
\newtheorem{remark}[theorem]{Remark}
\newtheorem{notation}[theorem]{Notation}
\newtheorem{assumption}{Assumption}

\usepackage{graphicx}

\usepackage{amsmath,amssymb,stmaryrd}
\usepackage{latexsym}
\usepackage{bbm}
\usepackage[all]{xypic}
\usepackage{hyperref}[]
\usepackage{bbm}

\newcommand{\hide}[1]{}

\newcommand{\nat}{\mathbb{N}}

\newcommand{\ol}[1]{\overline{#1}}

\newcommand{\wt}[1]{\widetilde{#1}}

\newcommand{\B}{\mathcal{B}}
\newcommand{\BA}{\mathit{BA}}
\newcommand{\C}{\mathcal{C}}
\newcommand{\D}{\mathcal{D}}
\newcommand{\Dsub}{\D_{\leq 1}}
\newcommand{\Deq}{\D_{=1}}
\newcommand{\Rel}{\mathit{Rel}}
\newcommand{\Set}{\mathit{Set}}

\newcommand{\BAlg}{B\dash\mathit{Alg}}

\renewcommand{\P}{P}
\newcommand{\Pom}{\P_{\omega}}

\newcommand{\Id}{\mathit{Id}}

\newcommand{\T}{T}
\newcommand{\olT}{\overline{\T}}
\newcommand{\wtT}{\widetilde{\T}}

\renewcommand{\L}{\mathcal{L}}

\newcommand{\Kl}[1]{\mathit{Kl}#1}

\newcommand{\id}{\mathit{id}}

\newcommand{\tr}{\mathit{tr}}
\newcommand{\wttr}{\widetilde\tr}
\newcommand{\st}{\mathit{st}}
\newcommand{\dst}{\mathit{dst}}

\newcommand{\struct}[1]{\langle #1\rangle}

\newcommand{\To}{\Rightarrow}
\newcommand{\too}{\longrightarrow}

\newcommand{\ladj}{\dashv}

\newcommand{\dash}{\mbox{-}}

\newcommand{\tick}{\surd}

\newcommand{\sem}[1]{[#1]}

\newcommand{\srngs}{\mathcal{S}}
\newcommand{\srngsom}[1]{(\srngs^{#1})_{\omega}}

\renewcommand{\lim}{\mathit{lim}}

\newcommand{\op}[1]{{#1}^{\mathit{op}}}

\newcommand{\bbtwo}{\mathbbm{2}}

\newcommand{\acal}{\mathcal{ A}}

\newcommand{\dcal}{\mathcal{ D}}

\newcommand{\pcal}{\mathcal{ P}}

\newcommand{\twobb}{\mathbbm{2}}

\newcommand{\Act}{\mathit{Act}}

\newcommand{\SLat}{\mathit{SLat}}
\newcommand{\diam}[1]{\langle #1\rangle}
\newcommand{\lsem}{\mathopen{[\![}}    
\newcommand{\rsem}{\mathclose{]\!]}}    
\renewcommand{\sem}[1]{\lsem #1 \rsem}

\renewcommand{\min}{\mathit{min}}

\newcommand{\Klom}[1]{\mathit{Kl}_{\omega}#1}

\newcommand{\wtgamma}{{\widetilde\gamma}}

\renewcommand{\P}{\mathcal{P}}
\newcommand{\eps}{\varepsilon}
\newcommand{\QSet}{Q_0}

\begin{document}
\title{Generic Trace Logics}
\author{Christian Kissig and Alexander Kurz\\ 
\small University of Leicester, Department of Computer Science\\
}
\maketitle

\abstract{We combine previous work on coalgebraic logic with the
  coalgebraic traces semantics of Hasuo, Jacobs, and Sokolova.}

\section{Introduction}

The coalgebraic approach to modal logic has been pursued
successfully over the last years. The basic ideas (see eg
\cite{moss:cl,pattinson:cml-j,schroeder:fossacs05,kurz:sigact06}),
are the following.
\begin{itemize}
\item A $T$-\emph{coalgebra}, consisting of a carrier $X$ and a
  `next-step' map $\xi:X\to TX$, represents a transition system. For
  example, with $\pcal X$ the set of finite subsets of $X$ and
  $\Act$ a set of actions, $X\to\pcal(\Act\times X)$ is a
  labelled transition system.
\item Any particular choice of $T$ yields a canonical notion of
  $T$-\emph{bisimilarity}. For example, for $X\to\P(\Act\times X)$
  we obtain the Milner-Park notion of bisimilarity~\cite{aczel:nwfs}
  whereas for $X\to\D(\Act\times X)$, with $\D X$ denoting the
  set of probability distributions on $X$, we obtain the
  notion of bisimilarity described in \cite{vink-rutt:prob-bisim}.
\item Moreover, for any choice of $T$, we can find a logic for
  $T$-coalgebras which is expressive (ie distinguishes non-bisimilar
  states) and comes with a complete calculus. These logics are modal
  logics in the sense that formulas are invariant under
  $T$-bisimilarity.
\end{itemize}

The work on coalgebraic logic so far is focused on $T$-bisimilarity.

\medskip In parallel, Jacobs and collaborators
\cite{Jacobs04TraceSemantics,HasuoJacobsSokolova06GenericTraceTheory,Hasuo08PhD} showed that coalgebras not only
provide a framework for bisimilarity, but also for trace semantics:
\begin{itemize}
\item A $(B,T)$-coalgebra $X\to BTX$ is now given wrt a `transition
  type' $T$ and a `branching type' $B$. For example, with $BX=\pcal X$
  and $TX=\{*\}+\Act\times X$, a $X\to\P(\{*\}+\Act\times X)$ is a
  non-deterministic automaton.
\item Different choices of $B$ yield different notions of trace
  semantics. With $B=\P$, the trace semantics of
  $X\to\pcal(\{*\}+\Act\times X)$ identifies states that accept the same
  language. With $B=\D$, the trace semantics of
  $X\to\D(\{*\}+\Act\times X)$ identifies states that accept the same
  (finite) traces with the same probabilities.
\end{itemize}

The work of Jacobs et al is build on several assumptions, which limit the
generality of the definition of trace semantics. For instance, it is not
possible to define the trace semantics of finitely branching transition
systems.

\medskip\noindent\textbf{Results } In this paper, we reconsider the
definition of trace semantics in the category of algebras for the
branching type $B$. This allows us to includes the often occuring finite
non-determinism and finitely graded branching.

Moreover we propose a generic definition of coalgebraic logics characterising
states up to trace equivalence. Our definition of trace logics is build upon a 
dual adjunction on the category of algebras for the branching type, and matches
the definition of coalgebraic modal logics for $T$-bisimulation. 

\medskip\noindent\textbf{Structure of the paper } After reviewing
material known from the literature, Section~\ref{sec:EM} introduces
trace semantcis in the category of Eilenberg-Moore algebras of the
monad $B$ describing the branching type. Section~\ref{MT-logic}
describes trace logics using the adjunction induced by the closed
structure a the commutative monad $B$. Section~\ref{sec:pred-lift}
explains how to define logics via predicate lifting, a notion known
set-coalgebras, which is adapted to our
setting. Section~\ref{sec:gen-logic} introduces the notion of a
generic trace logic and uses it to prove a particular instance to be
sound, complete, and expressive.

\medskip\noindent\textbf{Acknowledgements} We would like to thank
Ichiro Hasuo and Bart Jacobs.

\section{Two Examples}\label{sec:examples}
\label{S:Examples}

Consider $\gamma:X\to\Pom(\{*\}+\Act\times X)$. $(X,\gamma)$ is a
finitely non-deterministic automaton. Indeed, with $1$ as $\{*\}$ and
$+$ as (disjoint) union, we read $(a,x')\in\gamma(x)$ as $x$ can input
$a$ and go to $x'$ and we read $*\in\gamma(x)$ as $x$ is an accepting
state.

Now consider a logic
\begin{equation}\label{equ:exle1-syntax}
\phi ::= 0\mid \tick \mid \phi\vee\phi\mid \diam{a}\phi
\end{equation}
with compositional semantics
\begin{align}
x\not\Vdash 0 & \\
x\Vdash \tick & \Leftrightarrow \ *\in\gamma(x)\\
x\Vdash \phi\vee\psi \ & \Leftrightarrow \ x\Vdash\phi\ \textrm{or} \
x\Vdash\psi \\
x\Vdash\diam{a}\phi \ & \Leftrightarrow \ (a,x')\in\gamma(x)\ \textrm{and}
\ x'\Vdash\phi \label{equ:exle1-semantics-diam}
\end{align} 
and as axiomatisation the usual laws for falsum (0) and disjunction
($\vee$) plus the axioms
\begin{equation}\label{equ:exle1-axioms}
  \diam{a}0 = 0 \quad\quad \diam{a}(\phi\vee\psi) =
  \diam{a}\phi\vee\diam{a}\psi
\end{equation}
Note that this implies the typical axiom we would expect for trace
logics
\begin{equation}
  \label{Eq:Bifurcation}
\diam{a}(\diam{b}\phi\vee\diam{c}\psi) =
\diam{a}\diam{b}\phi\vee\diam{a}\diam{c}\psi
\end{equation}
Our development will not only provide a generic proof for the fact
that this logic is sound, complete and expressive, but also provide
conceptual explanations for why we can have falsum and disjunction,
but not negation and conjunction. 

\medskip To see that the interaction of the modal operators $\diam{a}$
with the propositional operators (0,$\vee$) is subtle, consider as a
second example $\gamma:X\to\D(\{*\}+\Act\times X)$ where $\dcal Y$ is the
set of finitely supported discrete probability distributions on
$Y$. $\gamma(x,*)\in[0,1]$ is the probability of terminating successfully
and $\gamma(x,a,x')\in[0,1]$ is the probability of continuing with $a$
and transiting to $x'$. Two states $x,x'$ are trace equivalent if
(inventing an adhoc notation similar to the logic above)
\begin{equation}\label{equ:exle:prob}
x\Vdash p\cdot\diam{a_0}\ldots\diam{a_n}\tick \ \Leftrightarrow \
x'\Vdash p\cdot\diam{a_0}\ldots\diam{a_n}\tick
\end{equation}
which we read as stating that the probability of $x$ (and $x'$) to
terminate successfully after the sequence $a_0\ldots a_n$ is $p$.

The notation in (\ref{equ:exle:prob}) indicates that there must be a
definition of logic, semantics, axiomatisation paralleling the example
of non-determinstic automata and we will show how to obtain in a
systematic fashion from the functors involved.

\section{Preliminaries}
\label{sec:preliminaries}

\subsection{Monads, Algebras and Coalgebras}

\begin{definition}
  A \emph{coalgebra} for an endofunctor $T$ on a category $\C$ is a
  morphism $\gamma:X\to TX$ for an object $X$ of $\C$, that we call
  $\gamma$'s domain. A $T$-coalgebra morphism between coalgebras
  $\gamma:X\to TX$ and $\delta:Y\to TY$ is a morphism $f:X\to Y$ such
  that $Tf\circ \gamma=\delta\circ f$ commutes. Dually, a $T$-algebra
  is an arrow $\alpha:TX\to X$. 
\end{definition}

\begin{definition}
  A \emph{monad} on $\Set$ is an endofunctor $B:\Set\to\Set$ with
  natural transformations $\eta:\Id\To B$ and $\mu:BB\To B$ such that
  $\mu\circ\eta_T=\id_T=\mu\circ T\eta$ and $\mu\circ\mu_T=\mu\circ
  T\mu$. If $B$ preserves filtered colimits, the monads is called
  \emph{finitary}.
\end{definition}

\begin{example}[finitary monads]  
  \label{Ex:Monads}
  \begin{enumerate}
  \item The finite powerset $\Pom$, equipped with the singleton map
    $\{(-)\}$ and set-union. 
  \item The bag functor $\B$ takes a set $X$ to the set
    $(\nat^X)_{\omega}$ of its finite multisets, and functions $f:X\to
    Y$ to multiset-functions $\B f:\B X\to\B Y$ taking multisets
    $m\in(\nat^X)_{\omega}$ to $\lambda y.\sum_{x\in
      f^{-1}(y)}m(x)$. 
  \item A (sub-)distribution of a set $X$ is a function $d:X\to[0,1]$
    such that $\sum_{x\in X}d(x)=1$ ($\sum_{x\in X}d(x)\leq 1$). The
    (sub-)distribution functor $\Deq$ ($\Dsub$) takes a set $X$ to the
    set of its (sub-)distributions, and functions $f:X\to Y$ to
    $\lambda m.\lambda y.\sum_{x\in f^{-1}(y)}m(x)$. For the sake of a
    brevity we write both, $\Deq$ and $\Dsub$, as $\D$ when it is
    clear from context, which functor we mean.

    For each $X$ we can define functions
  \[
  \begin{array}{lll}
	\mu_X(d'\in\D^2 X)(x):=\sum_{d\in\D X}d'(d)\cdot d(x)  & ~~ &
	\eta_X(x):=\lambda y.\left\{\begin{array}{ll}1&\mbox{ if }y=x\\0&\mbox{ otherwise}\end{array}\right.
  \end{array}
  \]
  $\mu$ and $\eta$ are transformations natural in $X$ and form with
  $B$ a monad.
\item All of the above are examples of functors which take a set $X$ into
  the set $\srngsom{X}$ of evaluations of $X$ into a semiring $\srngs$
  with finite support, and functions $f:X\to Y$ into functions
  $\srngsom{X}\to\srngsom{Y}$ such that $m\in\srngsom{X}\mapsto\lambda
  y.\sum_{x\in f^{-1}(x)}m(x)$.  For $\Pom$ the semiring is the boolean
  algebra $\struct{\{\top,\bot\},\wedge,\vee,\top,\bot}$, and for $\B$
  the semiring are the natural numbers $\struct{\nat,+,*,0,1}$.
\item If we take for $\srngs$ the real numbers with addition and
  multiplication, then the category of algebras for the semiring monad
  is (isomorphic to) the category of vector-spaces. See
  Semadeni~\cite{semadeni:monads} for more on this perspective. More
  generally, if the semiring does not happen to be a field, the
  category of algebras for the monad is known as the category of
  modules for the semiring.
\item Another example of a semiring monad uses the min-semiring
  $\struct{\nat\cup\{\infty\},\min,+,\infty,0}$ of natural numbers
  augmented with a top element, $\infty$, with an idempotent additive
  operation, $\min$, and a commutative multiplicative operation, $+$,
  such that $\infty$ is neutral wrt $\min$ and $0$ wrt $+$, and $0$
  absorbs wrt $\min$. 
\item Another example of semiring monads can be found in the weighted
  automata of Rutten~\cite{Rutten03Weighted}, where the stream
  behaviour is an instance of the finite trace semantics presented in
  this paper.
  \end{enumerate}
\end{example}


An (Eilenberg-Moore-) algebra for a monad $B$ is an algebra for the
functor $B$ satisfying additionally $\alpha\circ\mu_X=\alpha\circ
B\alpha$ and $\alpha\circ\eta_X=id_X$.
The algebras for a monad $B$ form a category, the Eilenberg-Moore
category $\BAlg$.  $U:\BAlg\to\C$ maps an algebra to its carrier.  $U$
has a left adjoint $F$ and we write $\eta:\Id\to UF$ and
$\eps:FU\to\Id$ for the unit and counit of the adjunction. Recall that
$UF=B$ and $F\eps_{UX}=\mu_X$.

Each monad admits and initial and a final $B$-algebra, respectively
$\struct{B\emptyset,\mu_{\emptyset}B^2\emptyset\to B\emptyset}$ and
$\struct{\{*\},(\lambda _.*):B\{*\}\to\{*\}}$. Synonymously, we denote
by $1$ a singleton set, when the domain ($\Set$ or $\BAlg$) is clear
from context.

For our definition of generic trace logics, it may be useful when
$\BAlg$ is \emph{closed} in the sense that homsets in $\BAlg$ have
$B$-algebra structure themselves.  Kock~\cite{Kock70MonadsOnSMC}
showed that this is true for commutative monads.

\begin{definition}[Strength Laws]
  A strength law for a monad $B$ is a transformation
  $\st_{X,Y}:=BX\times Y\to B(X\times Y)$ natural in $X$ and $Y$ and
  commutes with the monad's unit and multiplication law such that
  $\st_{X,Y}\circ(\eta_X\times\id_Y)=\eta_{X\times Y}$ and
  $\mu_{X\times Y}\circ B\st_{X,Y}\circ\st_{BX,Y}=\st_{X\times
    Y}\circ(\mu_X\times\id_Y)$.

  A double strength law is a natural transformation given as the
  diagonal $\dst_{X,Y}:BX\times BY\to B(X\times Y)$ of $\mu_{X\times
    Y}\circ B\st_{Y,X}\circ\st_{X,BY}=\mu_{X\times Y}\circ
  B\st_{X,Y}\circ\st_{Y,BX}$, given it exists consistently.

  A monad is commutative if it has a double strength law.
\end{definition}

\noindent The proof of the following can be found in \cite{Kock70MonadsOnSMC}.

\begin{proposition}
  The Eilenberg-Moore category of a commutative monad is closed.
\end{proposition}

\subsection{The Kleisli Construction and Functor Liftings}

\begin{definition}[Kleisli-Categories]
  \label{D:KleisliCats}
  The Kleisli-category $\Kl B $ of a monad $B$ on $\C$ has as objects
  the objects of $\C$ and arrows $f:X\to Y$ are the arrows $f:X\to BY$
  in $\C$. The identity is given by $\eta:X\to BX$ and composition of
  $f:X\to Y$ and $g:Y\to Z$ in $\Kl B $ is given by $g\circ
  f:=\mu_Z\circ Bg\circ f$.

  The adjunction $F'\ladj U':\C\to\Kl B $ is defined such that for
  all sets $X$, $F'X:=X$, all functions $f:X\to Y$ in $\Set$,
  $F'f:=\eta_Y\circ f$, and for all objects $X$ in $\Kl B $, $U'X:=BX$
  and for all morphisms $f:X\to Y$, $U'f:=\mu_Y\circ Bf$.
\end{definition}

\begin{example}
  \begin{enumerate}
  \item The Kleisli-category for the powerset monad $\P$ is $\Rel$,
    the category of sets as objects and relations as morphisms.
    \item The Kleisli-category for the semiring monad $\srngsom{(-)}$ is the category of free (left) modules for the semiring $\srngs$.
  \end{enumerate}
\end{example}

A coalgebra $\gamma:X\to BTX$ in $\Set$ is a morphisms $X\to TX$ in
$\Kl B $. In order to exhibit $\gamma$ as a coalgebra in $\Kl B $ and
to have coalgebra morphisms, one defines the lifting of
$\Set$-functors $T$ to $\Kl B $. The lifted functor $\olT$ makes
$FT=\olT F$ commute.  The existence of the functor lifting is
equivalent to the existence of a distributive law.

\begin{definition}[Distributive Laws]
  A distributive law for a monad $B$ and a functor $T$ is a natural
  transformation $\pi:TB\To BT$ such that $\pi\circ T\eta=\eta_T$ and
  $\pi\circ T\mu=\mu_T\circ B\pi\circ\pi_B$ commute.
\end{definition}

\begin{example}
  \label{Ex:DistrLaws}
  Let $T(-):=\{*\}+\Act\times (-)$ be a $\Set$-functor for a fixed set
  $\Act$. With each of the monads in Example~\ref{Ex:Monads} $T$ has a
  distributive law.
  \begin{enumerate}
    \item $\pi:T\P\to \P T$: $\pi_X(*):=\{*\}$, $\pi_X(a,Y\subseteq X):=\{(a,x)\mid x\in Y\}$.
    \item $\pi:T\B\to\B T$: $\pi_X(*):=\eta_{\{*\}+\Act\times X}(*)$, and
      $\pi_X(a,m)(a,x):=\{(a,x)\mapsto m(x),(b,x)\mapsto 0,*\mapsto 0\mid a\in\Act,b\in\Act,b\neq a,x\in X\}$
    \item $\pi:T\D\to\D T$: $\pi_X(*):=\eta_{\{*\}+\Act\times X}(*)$, and
      $\pi_X(a,d):=\{(a,x)\mapsto d(x),(b,x)\mapsto 0,*\mapsto 0\mid a\in\Act,b\in\Act,b\neq a,x\in X\}$ where $\D\in\{\Dsub,\Deq\}$
  \end{enumerate}
\end{example}

\begin{definition}[Functor Lifting by Distributive Law]
  Given a distributive law $\pi:TB\to BT$ we can define $\olT$ on
  objects $\olT X:=\T X$ and on morphisms $\olT(f:X\to
  Y):=\pi_Y\circ\T f$
\end{definition}

There is a full and faithful functor $K:\Kl B \to\BAlg$ mapping $X$ to
the free algebra over $X$, see~\cite{MacLane98Categories}. In other
words, we can think of $\Kl B $ as the full subcategory of $\BAlg$
consisiting of the free algebras.

\section{Coalgebraic Logic for Trace Semantics}
\label{S:TraceLog}

In this section we show how to set up trace logics in a coalgebraic
framework. But first we review some basic of coalgebraic logic (more
can be found in \cite{kurz:sigact06}) and the fundamentals of generic
trace semantics \cite{Jacobs04TraceSemantics}.

\subsection{A Brief Review of Logics for $T$-Bisimilarity }
\label{S:CoalgLogic}

Suppose we are looking for a logic for $T$-coalgebras built upon
classical propositional logic. Such a logic would be based on Boolean
algebras which precisely capture the axioms of propositional logic.
Then, in the same way as $T$ is a functor $\Set\to\Set$ on the models
(coalgebras) side, the logic will contain modalities given in terms of 
a functor $L:\BA\to\BA$ on the category $\BA$ of Boolean algebra. 
The situation is depicted in

\begin{equation}\label{E:Duality}
  \xymatrix{
    \Set\ar@/^.5pc/[rr]^{Q}\POS!R(-.7),\ar@(ul,dl)_{T} & \bot
    & \op\BA\ar@/^.5pc/[ll]^{S}\POS!R(.7),\ar@(ur,dr)^{L}
  }
\end{equation}

$Q$ contravariantly takes sets $X$ to their powersets $2^X$ and 
$S$ maps a Boolean algebra to the set of
maximal consistent theories (ultrafilters). For example, if $T=\pcal$
we may define $L$ by saying that $LA$ is the Boolean algebra generated
by $\Diamond\phi, \phi\in A$, modulo the axioms
\begin{equation}\label{equ:K-axioms}
\Diamond 0=0 \quad\quad
\Diamond(\phi\vee\psi)=\Diamond\phi \vee \Diamond\psi
\end{equation} 
Note how this definition of $L$ captures the usual modal logic for
(unlabelled) transition systems. The semantics of the logic is given
by a map
\begin{equation}\label{equ:delta}
  \delta_X:LQX\to QTX
\end{equation}
In the example we define $\delta_X(\Diamond\phi)=\{\psi\in \pcal TX
\mid \phi\cap\psi\not=\emptyset\}$ in order to capture that
$\Diamond\phi$ holds if the set `of successors' $\psi$ satisfies
$\phi\cap\psi\not=\emptyset$. Finally, $(L,\delta)$ gives rise to a
logic in the usual sense as follows. The set of formulas of the logic
is the carrier of the initial $L$-algebra. The semantics of a formula
wrt to a coalgebra $X\to TX$ is given by the unique homomorphism from
the initial $L$-algebra $LI\to I$ as in:
\begin{equation}\label{equ:semantics}
\xymatrix{
LI\ar[d]_{L(\sem{\cdot})}\ar[rr] && I\ar[d]^{\sem{\cdot}}\\
LQX\ar[r]^{\delta_X} & QTX \ar[r]^{Q\gamma} & QX
}
\end{equation}

\begin{theorem}
  Any $(L,\delta)$ with $\delta$ as in (\ref{equ:delta}) gives rise to
  a logic for $T$-coalgebras. The semantics $\sem{\cdot}$ as in
  (\ref{equ:semantics}) is invariant under $T$-bisimilarity. The logic
  is expressive for (finite) coalgebras, if $\delta_X$ is onto for
  (finite) $X$ and the equational logic given by the axioms defining
  $L$ is complete if $\delta_X$ is injective for all $X$. 
\end{theorem}
\noindent Suppose we are given $T$, how can we find a logic
$(L,\delta)$? Two answers:
\begin{remark}\label{rem:generic-logics}
\begin{enumerate}
\item Moss~\cite{moss:cl} takes $LA$ to be the free $\BA$ generated by
  $TU\!A$ where $U\!A$ is the underlying set of $A$. A complete
  calculus has been given in \cite{kkv:aiml08}.
\item The standard modal logic for $T=\pcal$ above arises from
  $LA=QTSA$ on finite $A$ and extending continuously to all of $\BA$
  \cite{kurz-rosi:calco07}. It is always complete.
\end{enumerate}
\end{remark}
Both logics are expressive. A detailed comparison has been given in
\cite{KurzLeal09Equational}.

\subsection{A Brief Review of Finite Trace Semantics}
\label{S:FinTraceSem}

\textbf{The basic construction } Consider a coalgebra $X\to BTX$, the
running example being $B=\pcal$ and $TX=\{*\}+\Act\times X$ as
discussed in Section~\ref{sec:examples}. The set of traces will be
the carrier of the initial $T$-algebra given by the colimit (or union) of the
sequence
\begin{equation}\label{equ:traces}
\xymatrix{
\emptyset\ar@{^(->}[rr]^{\emptyset} &&
T\emptyset\ar@{^(->}[rr]^{\T\emptyset} &&
\T^2\emptyset\ar@{^(->}[r] & 
\cdots  &  T^\omega\emptyset}
\end{equation}
In the example $T^n\emptyset=\{a_1\ldots a_n\mid a_i\in\Act\}$ and
$T^\omega\emptyset=\Act^*$, ie the set of finite words over $\Act$.
The set of traces of length $n$ will be given by a map 
\begin{equation}\label{equ:trace-map}
\tr_n:X\to BT^n\emptyset
\end{equation} In the example, $\tr_n(x)$ is the set of traces of
length $n$ that lead from $x$ to an accepting state. To compute it, we
need the following ingredients.
\begin{assumption}\label{ass:B}\ 
\begin{itemize}
\item a map $\mu_X:BBX\to BX$ (for this we assume that $B$ is a monad)
\item a map $\pi_X:TBX\to BTX$ (for this we assume that $\pi$ is a distributive
  law)
\item an algebra morphism $e:A\to F\emptyset$ from any $B$-algebra $A$
  into $F\emptyset$.\footnote{This means that we assume from hereon
    $B\emptyset\not=\emptyset$. Also note that in all our examples $B$
    is a commutative monad, hence $B\emptyset\not=\emptyset$ implies
    $B\emptyset=1$, so that $F\emptyset$ is the final algebra.}
\end{itemize}
\end{assumption}

\noindent The maps $\tr_n$ then arise from taking $n$ steps of
$\gamma$, eg in the case $n=2$, as
\begin{equation*}
  \xymatrix{
    X \ar[r]^{\gamma\ \ } & BTX \ar[r]^{BT\gamma\ \ } & BTBTX\ar[r]^{BTBTe} & 
    BTBTB\emptyset\ar[r]^{p} & BBBTT\ar[r]^{m}\emptyset  & BT^2\emptyset}
\end{equation*}
($p$ stands for 3 applications of $\pi$ and $m$ for 2 applications of
$\mu$.)

\begin{definition}\label{def:tr-sem}
  Two states $x,y\in X$ of a coalgebra $X\to BTX$ are trace equivalent
  if $\tr_n(x)=\tr_n(y)$ for all $n<\omega$.
\end{definition}

For the purposed of the current paper, we consider this the essence of
the trace semantics of
\cite{HasuoJacobsSokolova06GenericTraceTheory}. But
\cite{HasuoJacobsSokolova06GenericTraceTheory} do much more and, in
particular, they show that under additional assumptions the trace
semantics can be given by a final coalgebra in the Kleisli category.

\medskip\noindent\textbf{Trace semantics in the Kleisli category }
\cite{HasuoJacobsSokolova06GenericTraceTheory} show not only that the
ingredients of a monad $B$ and a distributive law $TB\to BT$ give rise
to trace semantics, they also show that it can be elegantly formulated
in the so-called Kleisli category of the monad $B$ (see
Section~\ref{sec:preliminaries}). The objects in the Kleisli category
are the same as in $\Set$, but arrows $X\to Y$ in $\Kl B $ are maps
$X\to BY$ in $\Set$. In case of the powerset functor $B=\pcal$,
$\Kl B $ is the category of sets with relations as arrows.

The distributive law $TB\to BT$ gives rise to a lifting of
$T:\Set\to\Set$ to $\olT:\Kl B \to\Kl B $. The definition of $\tr_n$
can then be defined inductively as
\begin{equation}\label{equ:def-tr}
\tr_{n+1}=\olT(\tr_n)\circ\gamma
\end{equation} 
where we assume a morphism $\tr_0:X\to 0$ in the base case. The
following diagram illustrates the above definition

\begin{equation}
\label{equ:traces-in-Kl}
\xymatrix{
& 
X\ar[r]^{\gamma}\ar@/_.5pc/[ld]_{\tr_0}\ar@/^1pc/[rd]_{\tr_n}\ar@/^1pc/[rrd]_{\tr_{n+1}}& 
\olT X\ar@/^1pc/[rd]^{\olT\tr_n} &&\\
\emptyset &
\cdots &
\olT^n\emptyset &
\olT^{n+1}\emptyset &
\cdots\\
}
\end{equation}

\hide{
in $\Kl B $
\begin{equation}\label{equ:traces-in-Kl}
  \xymatrix{
    \emptyset\ar[r]^{\emptyset} &
    \olT\emptyset\ar[r]^{\T\emptyset} &
    \olT\olT\emptyset & 
    \cdots & \olT\olT^n\emptyset & \cdots \\
    X \ar[r]^{\gamma} \ar[u]^{\tr_0} & \olT X \ar[r]^{\olT \tr_0}
    \ar[u]^{\olT \tr_0}  & \olT\olT X \ar[u]^{\olT\olT\tr_0}& 
      \cdots & \olT^n X 
    \ar[u]^{} & \cdots
  }
\end{equation}
}

Furthermore, under conditions for which we refer to
\cite{HasuoJacobsSokolova06GenericTraceTheory}, the final
$\olT$-coalgebra $Z$ exists.\footnote{Moreover,
  \cite{HasuoJacobsSokolova06GenericTraceTheory} prove the beautiful
  result that show that the final $\olT$-coalgebra is given by the
  initial $T$-algebra with the carrier $T^\omega\emptyset$ as in
  (\ref{equ:traces}).}  Therefore, with the notation of
Definition~\ref{def:tr-sem}, there is a map $\tr:X\to BZ$ with the
property $$\tr(x)=\tr(y) \Leftrightarrow \tr_n(x)=\tr_n(y)$$ for all
$n<\omega$. Thus, the trace semantics via the final coalgebra (if it
exists) in the Kleisli-category is equivalent to the one of
Definition~\ref{def:tr-sem}. The advantage of the trace semantics via
the final coalgebra in the Kleisli-category is that it gives a
coinductive account of trace semantics. The disadvantage is that it
excludes some natural examples such as \emph{finite} powersets or
multisets. The next section shows that these examples can be treated
via final coalgebras if we move from the Kleisli-category to the
category of algebras for the monad.

\subsection{Trace Semantics in the Eilenberg-Moore
  Category}\label{sec:EM}

In this section we propose to move the trace semantics from the
Kleisli-category $\Kl B $ to the category $\BAlg$ of
Eilenberg-Moore-algebras. There are at least two reasons why this of
interest. The first is that the duality we will exploit for the logic
takes place in $\BAlg$. The second is that, in general, the limit of
Diagram~(\ref{equ:traces-in-Kl}) is not a free $B$-algebra and hence
not in $\Kl B $, but it always exists in $\BAlg$.  

Let $K$ denote the functor which embeds $\Kl B$ into $\BAlg$. Our
first task is to extend $\olT:\Kl B\to\Kl B$ to $\wtT:\BAlg\to\BAlg$
so that $\wtT K\cong K\ol T$ (hence $\wtT F\cong F T$).

\begin{equation}
  \xymatrix{
    \Klom B\ar@/_1pc/@{^(->}[rr]_{J}\ar@(ul,ur)^{\olT'} &
    \Kl B\ar@/^.5pc/@{^(->}[r]^{K}\ar@(ul,ur)^{\olT} &
    \BAlg\ar@(ul,ur)^{\wtT}
  }
\end{equation}
On the full subcategory of free algebras we can define $\wtT FX=K\olT
X=FTX$. To extend this to arbitrary algebras $A$ recall first that any
$A\in\BAlg$ is a coequaliser of $FU\eps_A,\eps_{FUA}:FUFUA\to FUA$. We
then define $\wtT A$ as the coequaliser of $\wtT FU\eps_A$ and $\wtT
\eps_{FUA}$. It can be shown that $\wtT$ is the left Kan-extension of
$K\olT$ along $K$. 

\begin{example}\label{exle:wtT}
  Let $B=\Pom$ and $T=\{\tick\}+\Act\times\Id$. Then $\wtT A\cong
  F1+\Act\cdot A$. Indeed, by definition, we have $\wtT FX= FTX \cong
  F1+\Act\cdot FX$. Now the claim follows from the fact that the
  functor $F1+\Act\cdot\Id$, being a coproduct, preserves
  coequalisers.
\end{example}

It is convenient for us to make the following assumptions.

\begin{assumption}\label{ass:fin}
  $\B:\Set\to\Set$ is a finitary commutative monad with
  $B\emptyset\not=\emptyset$ and $T:\Set\to\Set$ is a finitary functor
  with a distributive law $TB\to BT$.
\end{assumption}

\begin{remark}\label{rmk:wtT-fb} 
  If $\B$ and $T$ are finitary, then $\wtT$ is determined by finitely
  generated free algebras, or, in other words, $\wtT$ preserves sifted
  (hence filtered) colimits \cite{arv} and falls within the framework
  considered in \cite{kurz-rosi:strcompl,vele-kurz:fb}. For a functor
  $H:\acal\to\acal$ on a finitary algebraic category $\acal$ to be
  strongly finitary means that $H$ is determined by its action on
  finitely generated free algebras. More formally, $H$ is a left
  Kan-extension of $HK$ along $K$ where $K$ is the inclusion
  $\acal_0\to\acal$ of the full subcategory $\acal_0$ of finitely
  generated free algebras. A pleasant consequence is that all concrete
  calculations of some $HA$ can be restricted to the case $A=Fn$,
  where $F$ is the left adjoint of the forgetful functor
  $\acal\to\Set$ and $n$ is finite. This will be exploited in the
  following for $\acal=\BAlg$. Other consequences of our assumption
  then are:
\begin{itemize}
\item $F\emptyset$ is the initial and final object of $\BAlg$.
\item The final $\wtT$ sequence converges after $\omega$ steps.
\end{itemize}
\end{remark}

In a second step, we can now map a coalgebra $\gamma:X\to BTX$ (ie
$\gamma:X\to\olT X$) to $\wtgamma:FX\to \wtT FX$ (ie $\wtgamma:KX\to
K\olT X$). Thus $\wtgamma$ is a coalgebra for a functor
$\wtT:\BAlg\to\BAlg$. Moreover we observe that we can factor
$\tr_n:X\to BT^n\emptyset$ from Diagram~(\ref{equ:traces-in-Kl}) as
$$ \tr_n : X\to BX\cong UFX
\stackrel{U\wttr_n}{\too}\wtT^nF\emptyset\cong BT^n\emptyset$$ where
we define $\wttr_0$ via $e$ as in Assumption~\ref{ass:B} and
$\wttr_{n+1}=\wtT\wttr_n\circ\wtgamma$. Let us summarise this in a
definition and a proposition.

\begin{definition}
  Recall Assumption~\ref{ass:fin}. For any coalgebra $\alpha:A\to\wtT
  A$ we define the trace semantics as follows. First, $\wttr:A\to
  F\emptyset$ is given by finality; then, inductively
  $\wttr_{n+1}=\wtT\wttr_n\circ\wtgamma$. This defines a cone on the
  final $\wtT$-sequence so we can define the trace semantics
  $\wttr:A\to Z$, where $Z\to\wtT Z$ is the final
  $\wtT$-coalgebra. For a coalgebra $\gamma:X\to BTX$ we define
  $\tr:X\to UZ$ as $U\wttr\circ\eta_X$, where $\wttr$ is the trace
  semantics of $\wtgamma:FX\to\wtT FX$.
\end{definition}

To emphasise that this definition agrees with the one of the previous
subsection we state

\begin{proposition}
  Consider $\gamma:X\to BTX$ and $\wtgamma:FX\to FTX= \wtT
  FX$. Then $U\wttr_n\circ \eta_X=\tr_n$. 
\end{proposition}

\noindent Thus, $Z$ and $\wttr$ and $\tr$ are just a convenient way to
talk about the maps $\tr_n$ for all $n\in\mathbb{N}$
simultaneously. In particular, we have now again a coinductive account
of trace semantics. This technique will give, for example, a short and
conceptual proof of Theorem~\ref{thm:main}. Under
Assumption~\ref{ass:fin}, and if the final $\olT$-coalgebra of
\cite{HasuoJacobsSokolova06GenericTraceTheory} exists, then both the
trace semantics in $\Kl B $ and the trace semantics in $\BAlg$ are
equivalent as both boil down to Definition~\ref{def:tr-sem}. (Of
course, this is due to the fact that the definition of $\wtT$ extends
to all algebras the lifting $\olT$ of $T$ to $Kl(B)$.)

\begin{remark}\label{rmk:fin-sem-BAlg} If $B\emptyset\not=0$ then the sequence
  $(\wtT^nF\emptyset)_{n<\omega}$ is the finitary part of the final
  $\wtT$-sequence in $\BAlg$. Moreover, it follows from
  Remark~\ref{rmk:wtT-fb} that if $B$ is finitary, then the
  $\omega$-limit $(\wtT^\omega F\emptyset)$ of the final sequence is
  the final $\wtT$-coalgebra. To summarise, in addition to the
  explanation of trace semantics as a final semantics in the
  Kleisli-category as in
  \cite{HasuoJacobsSokolova06GenericTraceTheory}, we can also give a
  final semantics in the Eilenberg-Moore category. These two
  approaches are slightly different, for example, the approach of
  \cite{HasuoJacobsSokolova06GenericTraceTheory} works for $B=\P$ but
  not for $B=\Pom$, whereas for us it is more natural to work with
  $B=\Pom$ as we then have algebras with a finitary signature.
\end{remark}

\begin{example}\label{exle:B=Pom:2}
  Consider $B=\Pom$, $T=\{\tick\}+\Act\times\Id$. Then
  $\wtT(FX)=F\{\tick\} + \Act\cdot FX$. We can identify $F\emptyset$
  with $\{\emptyset\}$ and $\wtT^n(F\emptyset)$ with $\Pom(1+\Act
  +\ldots \Act^n)$. Thus, elements of $\wtT^n(F\emptyset)$ are finite
  sets of finite words $\langle a_1\ldots a_i\rangle$, $i\le n$. As
  $F\emptyset$ is initial and final in $\BAlg$, the
  $\wtT^n(F\emptyset)$ are part of the initial and of the final
  $\wtT$-sequence. The projections
  $p^{n+1}_n:\wtT^{n+1}(F\emptyset)\to\wtT^{n}(F\emptyset)$ are
  finite-union-preserving maps determined by acting as the identity on
  singletons $\{\langle a_1\ldots a_i\rangle\}$ for $i\le n$ and
  sending $\{\langle a_1\ldots a_{n+1}\rangle\}$ to $\emptyset$. The
  embeddings $e^n_{n+1}:\wtT^n(F\emptyset)\to\wtT^{n+1}(F\emptyset)$
  are given by the obvious inclusions. Note that $p^{n+1}_n\circ
  e^n_{n+1}= \id_n$. The colimit of the initial $\wtT$-sequence
  $(e^n_{n+1})_{n<\omega}$ is given by all finite subsets of
  $\Act^*=\coprod_{n<\omega}\Act^n$. The limit of the final
  $\wtT$-sequence $(e^n_{n+1})_{n<\omega}$ is given by all subsets of
  $\Act^*$. Note that although all approximants $\wtT^n(F\emptyset)$
  are free algebras, the limit $\P(\Act^*)$ is not free in $\BAlg$ and
  hence does not appear in $\Kl(\Pom)$.
\end{example}

\subsection{Logics for Finite $B$-Traces }
\label{MT-logic}

We develop logics for $(B,T)$-coalgebras with a semantic invariant 
under trace equivalence in analogy to coalgebraic modal logic
for $T$-bisimulation.

Firstly we need a category carrying our logics. We have a number of
possible replacements for $\BA$ in Diagram~(\ref{E:Duality}):
distributive lattices for positive logic, Heyting algebras for
intuitionistic logic, complete atomic Boolean algebras for infinitary
logic. The minimal choice (without propositional operators) is $\Set$
itself as used for example by Klin in~\cite{klin:lics07}.

\begin{equation}\label{E:setsetop}
  \xymatrix{
  \Set\ar@/^.5pc/[rr]^{2^{(-)}} &
  \bot & 
  \Set^{op}\ar@/^.5pc/[ll]^{2^{(-)}}
  }
\end{equation}

In the above situation, $2$ takes the role of a schizophrenic
object. Analogously we may choose a $B$-algebra $\Omega$ to replace
$2$. In most examples we have considered, $F1$ is a suitable
choice, but for the moment we do not need to fix a choice.

\begin{notation}
  If $B$ is a commutative monad, we write $Q$ for the contravariant
  endofunctor $[-,\Omega]$ on $\BAlg$ where $\Omega$ is for now an
  arbitrary but fixed object of `truth values'.
\begin{equation}\label{E:balgbalgop}
  \xymatrix{
  \BAlg\ar@/^.5pc/[rr]^{Q=[-,\Omega]}	
  &\bot& 
  \op\BAlg\ar@/^.5pc/[ll]^{Q=[-,\Omega]}
  }
\end{equation}
$\QSet$ is the contravariant endofunctor
$U[F-,\Omega]=Set(-,U\Omega)$. We have $UQA=\QSet UA$.
\end{notation}

\begin{example}\label{exle:one}
  When $B=\Pom$, $\BAlg=\SLat$ is the category of (join)
  semi-lattices. For $\Omega$ we choose the two-element semi-lattice
  $F1=\twobb$, so that $[-,F1]$ takes a semi-lattice $A$ to the set of
  `prime filters' over $A$. For future calculations, we record some
  facts about semi-lattices. First, for finite $A$, there are
  order-reversing bijections
  \begin{equation}\label{exp-log}
\xymatrix{
A \ar@/^/[rr]^{\exp} & & [A,\twobb] \ar@/^/[ll]^{\log}
}
\end{equation}
given by $\exp(a)=\lambda b. \neg(b\le a)$ and
$\log(\phi)=\bigvee\phi^-$ where $\neg:\twobb\to\twobb$ is negation
and $\phi^- = \{a\in A\mid \phi(a)=0\}$. Another description of $\log$
goes as follows. Since $\phi$ preserves joins it has a right adjoint
$\phi^\sharp$ and $\log(\phi)=\phi^\sharp(0)$. Second, if $A=FX$ with
$X$ not necessarily finite, we have  the bijection
  \begin{equation}\label{equ:UQFX}
  UQFX=U[FX,\bbtwo]\cong\Set(X,2)\cong \P X
\end{equation}
which lifts to a semi-lattice isomorphism
\begin{equation}\label{equ:QFX}
  QFX\cong (\P X, \emptyset, \cup)
\end{equation}
mapping $\phi\in[FX,\twobb]$ to $\{x\in X\mid \phi(\{x\}) = 1\}$ and
$S\subseteq X$ to the unique $\phi$ with $\phi(x)=1\Leftrightarrow
x\in X$, or, equivalently, to $\lambda S'\in FX\,.\, S\cap
S'\not=\emptyset$ (where we use \eqref{equ:UQFX} to identify $S'$ with
a subset of $X$).  Taking now $X=n$ finite again, we
obtain
  \begin{equation}\label{setify0}
QFn\cong Fn.
\end{equation}
In this case it is more convenient to use $\exp$ and $\log$ to denote
the order-preserving bijections
  \begin{equation}\label{exp-log}
\xymatrix{
Fn \ar@/^/[rr]^{\exp} & & [Fn,\twobb] \ar@/^/[ll]^{\log}
}
\end{equation}
given by $\log(\phi)=\{i\in n\mid \phi(\{i\})=1\}$ and
$\exp(S)=\lambda S'. S\cap S'\not=\emptyset$ (where again we identify
elements $S, S'$ of $Fn$ with subsets $S,S'\subseteq n$).

One can check that $\exp(\exp(a)) = \lambda\phi.\phi(a)$. It follows
that $\exp\circ\exp: \Id\to QQ$ is the unit of the adjunction
\eqref{E:balgbalgop}, and, moreover, that the unit is an isomorphism
on finite semi-lattices.\footnote{This also follows from the fact that
  the adjunction \eqref{E:balgbalgop} restricts to an equivalence on
  finite semi-lattices \cite{johnstone:stone-spaces}.} In case of
$Fn\to QQFn$ we have for $S\subseteq n$ that $\exp(\exp(S))(\phi) =
\log(\phi)\cap S\not=\emptyset$. The inverse $QQFn\to Fn$ of $Fn\to
QQFn$ maps $u:[Fn,\twobb]\to\twobb$ to $\log(\log(u))=n\setminus
\{i\in n\mid \exists\phi\,.\,u(\phi)=0 \,\&\, \phi(\{i\})=1\}$.

We will also use that for finite semi-lattices coproducts and products
coincide, with 
\begin{equation}
  \label{equ:co-prod}
  \begin{array}{rcl}
    A+B & \to & A\times B\\
    a & \mapsto & (a,0)\\
    b & \mapsto & (0,b)\\
    a\vee b & \mapsfrom & (a,b) 
  \end{array}
\end{equation}
describing the isomorphism. \qed
\end{example} 

In Section~\ref{S:FinTraceSem} we have defined the finite trace
semantics of $\Set$-coalgebras $\gamma:X\to BTX$ as the
final coalgebra semantics of the lifted coalgebra $\gamma:FX\to\wtT FX$
in $\BAlg$.

Secondly we need a functor $L$ providing the modalities for our logics,
as in the following diagram.

\begin{equation}
  \label{E:balgbalgopTL}
  \xymatrix{
    \BAlg\ar@/^.5pc/[rr]^{Q}\POS!R(-.7)\ar@(ul,dl)_{\wtT} 
           &\bot& 
    \op\BAlg\ar@/^.5pc/[ll]^{Q}\POS!R(.7),\ar@(ur,dr)^{L}
  }
\end{equation}

In analogy to Section~\ref{S:CoalgLogic}, we develop finite trace
logics as the initial $L$-algebra $\L:LI\to I$ in $\BAlg$. Note that
under the assumptions of Remark~\ref{rmk:wtT-fb}, we have that $I$ is the
$\omega$-colimit of the initial $L$-sequence:
\begin{equation}
  \xymatrix{
  0\ar[r] & L0\ar[r] & L^20\ar[r] & \cdots
  }
\end{equation}

\begin{definition} 
  A trace logic is given by a functor $L:\BAlg\to\BAlg$ and a natural
  transformation $\delta: LQ\to Q\wtT$. Formulas of the logic are
  given by elements of the initial $L$-algebra. The semantics
  $\sem{\cdot}_\wtgamma$ wrt a $\wtT$-coalgebra $\wtgamma:FX\to\wtT FX$
  is given by initiality as in
\begin{equation}\label{equ:tr-semantics}
  \xymatrix{
    LI\ar[d]_{L\sem{\cdot}_\wtgamma}\ar[rr] && I\ar[d]^{\sem{\cdot}_\wtgamma}\\
    LQFX\ar[r]^{\delta_{FX}} & Q\wtT FX \ar[r]^{Q\wtgamma} & QFX
  }
\end{equation}
This induces the semantics $\sem{\cdot}_\gamma$ wrt a coalgebra
$\gamma:X\to BTX$ via
\begin{equation}\label{equ:sem-set-coalg}
  \xymatrix{
    UI\ar[r]^{U\sem{\cdot}_\wtgamma\ \ } & UQFX\ar[r]^\cong & \QSet UFX \ar[r]^{\QSet\eta_X} & \QSet X
  }
\end{equation}
\end{definition}

For future reference, we record that the semantics in terms of
$\gamma$ and $\wtgamma$ agree:

\begin{proposition}\label{prop:wtgamma}
  Let $\wtgamma:FX\to \wtT FX$ be the $\wtT$-coalgebra induced by the
  $(B,T)$-coalgebra $\gamma:X\to BTX$, that is,
  $\gamma=U\wtgamma\circ\eta_X$ with $\eta_X:X\to BX$ the unit of the
  monad $B$. Then $\sem{\phi}_\gamma(x)=
  \sem{\phi}_\wtgamma(\eta_X(x))$.
\end{proposition}

\begin{example}\label{exle:two}
  Continuing from Example~\ref{exle:one}, in order to describe the
  logic (\ref{equ:exle1-syntax}), we let $LA$ be the join-semilattice
  which is freely generated by $\tick$ and $\diam{a}\phi$ for
  $a\in\Act$ and $\phi\in A$, quotienting by
  (\ref{equ:exle1-axioms}). To describe $\delta_{FX}$ it is convenient
  to note that $QFX$ can be identified with the set of subsets of $X$
  as in \eqref{equ:QFX}
  and $Q\wtT FX= QFTX$ with the set of subsets of $TX$. It therefore
  makes sense to define
\begin{align*}
  \delta_{FX}:LQFX & \to Q\wtT FX \label{equ:tr-delta-exle1}\\
  \tick & \mapsto \{S\subseteq TX \mid *\in S\}\\
  \diam{a}\phi &\mapsto \{S \subseteq TX \mid \exists x(x\in\phi\
  \&\ (a,x)\in S)\}
\end{align*}
\end{example}

\begin{proposition}\label{prop:exle1}
$(L,\delta)$ of Example~\ref{exle:two}, together with
(\ref{equ:tr-semantics}), describes the same logic as (\ref{equ:exle1-syntax}) in
Section~\ref{sec:examples}.
\end{proposition}

\begin{proof}
For example, we calculate
$x\models\diam{a}\phi \ \Leftrightarrow \ %
\gamma(x)\in \{S \subseteq TX \mid \exists x'(x'\in\phi\ \&\ (a,x')\in
S)\}\Leftrightarrow \ %
\gamma(x)\in \delta_{FX}(\diam{a}\phi)\Leftrightarrow \ %
x\in QF\gamma(\delta_{FX}(\diam{a}\phi))\ \Leftrightarrow \ %
x\in\sem{\diam{a}\phi}$ where we use, respectively,
(\ref{equ:exle1-semantics-diam}), the definition of $\delta$, the
definition of $Q$, and (\ref{equ:tr-semantics}).
\end{proof}

\begin{theorem}\label{thm:main}
  Consider a functor $T:\Set\to\Set$, a monad $B$, and a distributive
  law $TB\to BT$. Any $(L,\delta)$ with $L:\BAlg\to\BAlg$ and
  $\delta_K:LQK\to QK\olT$ gives rise to a logic for $BT$-coalgebras
  invariant under $B$-trace semantics.
\end{theorem}

\begin{proof}
  For a given $\gamma:X\to BTX$ and formula $\phi$, we have to show
  that $\tr(x)=\tr(y)$ implies $x\Vdash\phi \Leftrightarrow
  y\Vdash\phi$. Expressing this in $\BAlg$, this amounts to
  $\wttr(\eta_X(x))=\wttr(\eta_X(y))$ only if $x\in\sem{\phi}_\wtgamma
  \Leftrightarrow y\in\sem{\phi}_\wtgamma$. But this is immediate from
  the initiality of the algebra of formulas as follows. Let
  $(Z,\zeta)$ be the final $\wtT$-coalgebra. 
\begin{equation}
  \xymatrix{
    LI\ar[d]_{L\sem{\cdot}_\zeta}\ar[rr] && I\ar[d]^{\sem{\cdot}_\zeta}\\
    LQZ\ar[r]^{\delta_{Z}}\ar[d]_{LQ\wttr} 
    & Q\wtT Z\ar[r]^{Q\zeta}\ar[d]_{Q\wtT\wttr} & QZ\ar[d]_{Q\wttr}\\
    LQFX\ar[r]^{\delta_{FX}} & Q\wtT FX \ar[r]^{Q\wtgamma} & QFX\\
  }
\end{equation}
Since morphisms from the initial algebra $LI\to I$ are uniquely
determined, we must have
$\sem{\cdot}_\wtgamma=Q\wttr\circ\sem{\cdot}_\zeta$.
\end{proof}

\hide{
\begin{remark}
  As in \cite{KupkeKurzPattinson04RLC} the logic will be complete if
  $\delta$ is in injective and expressive if $\delta$ is surjective
  and $L$ respectively preserves these properties.
\end{remark}

\begin{remark}
  there should be a remark analogous to Remark 1
\end{remark}
}

\subsection{Predicate Liftings}\label{sec:pred-lift}
\newcommand{\fp}{\mathsf{fp}} Whereas the previous section treats
logics from an abstract point of view, we are now going to see how to
describe them concretely using predicate liftings. First, we need to
extend the set-based notion of predicate lifting
\cite{pattinson:cml-j,schroeder:fossacs05} to coalgebras over $\BAlg$.

Suppose we have $L$ and $$LQ\to Q\wt T.$$ Using $\Id\to QQ$ from the
adjunction \eqref{E:balgbalgop} this gives us $$L\to LQQ\to Q\wt T
Q.$$ We will see below that $Q\wt T Q$ gives us predicate liftings,
but first we are going to show how to recover $LQ\to Q\wt T$ from
$L\to Q\wt T Q$. Write $$J:\Kl_\omega B\to \BAlg$$ for the inclusion of
the category of finitely generated free algebras into $\BAlg$. 

\begin{proposition}\label{prop:pred-lift}
  Let $L$ be determined by finitely generated free algebras as in
  Remark~\ref{rmk:wtT-fb}. Then there is a bijection between natural
  transformations $LQ\to Q\wtT$ and natural transformations $LJ\to
  Q\wtT QJ$.
\end{proposition}

\begin{proof}
  Given $\delta:LQ\to Q\wt T$ we obtain $\rho:LJ\to Q\wt T QJ$ as
  $\delta Q\circ L\eta$. Conversely, given $\rho$, we write $QA$ as a
  colimit $\phi_i:Fn_i\to QA$, which is preserved by $L$, and obtain
  $\delta$ via
  \begin{equation}\label{equ:deltarho}
    \xymatrix{
      QA & L QA \ar[r]^{\delta_A} & Q\wtT A\\
      Fn_i\ar[u]^{\phi_i} & L Fn_i \ar[u]^{L \phi_i}\ar[r]^{\rho_{Fn_i}} & Q\wtT QFn_i
      \ar[u]_{Q\wtT\check \phi_i} }
   \end{equation}
   where $\check \phi_i:A\to QFn_i$ is the adjoint transpose of $\phi_i$. To
   check that these two assignments are inverse to each other, we first
   note that the diagram \eqref{equ:deltarho} can be rewritten as
  \begin{equation}
    \xymatrix{
      & L QA \ar[r]^{\delta_A} & Q\wtT A\\ 
      L Fn_i \ar[ur]^{L \phi_i} \ar[r]^{L\eta} & 
      L QQFn_i \ar[u]_{L Q \check \phi_i} \ar[r]_{\delta QFn_i} & 
      Q\wtT QFn_i \ar[u]_{Q\wtT\check \phi_i} }
   \end{equation}
   where the triangle commutes because of the adjunction
   \eqref{E:balgbalgop} and the quadrangle commutes because of
   naturality. It follows that starting from $\delta$ and defining
   $\rho$, the original $\delta$ satisfies \eqref{equ:deltarho} and
   therefore agrees with the $\delta$ defined from $\rho$. Conversely,
   defining $\delta$ from $\rho$ in \eqref{equ:deltarho}, one can
   choose $A=QFn$, $n_i=n$ and $\check\phi=\id$, which shows that $\delta$
   determines the $\rho$ it comes from uniquely.
\end{proof}

We can interpret the proposition as follows. An element of $$\QSet
A=UQA$$ is a predicate on $A$. An element of $$[n,\QSet A]$$ is an
$n$-ary predicate on $A$. We have $[n,\QSet
A]\cong[Fn,QA]\cong[A,QFn]$ and find it useful to introduce the
following notation. We want to write $\phi$ for $n$-ary predicates and if
we want to make precise which of the three presentations we use, we
write 
\begin{equation}\label{equ:notation-predicates}
  \underline{\phi}\in [n,\QSet A] \quad\quad \phi = \hat \phi \in
  [Fn,QA]\quad\quad \check\phi \in [A,QFn].
\end{equation}

\medskip\noindent Next we show how elements $l\in LFn$ are $n$-ary
modal operators. Given an $n$-ary predicate $\phi$ on $A$, the `modal
operator' $l$ induces an predicate on $\wtT A$ as follows.
\begin{equation}\label{equ:tr-pred-lift}
\xymatrix{
\wtT A\ar[r]^{\wtT(\check a)} & \wtT QFn \ar[r]^{\rho_{Fn}(l)} & \Omega
}
\end{equation} 
This shows that the meaning of the modal operator $l\in LFn$ is fully
determined by the image $\rho_{Fn}(l)\in Q\wtT QFn$. We turn this
observation into a definition.

\begin{definition}
  Elements of $Q\wtT QFn$ are called $n$-ary predicate liftings. Each
  $\lambda\in Q\wtT QFn$ induces a natural transformation
\begin{equation}\label{equ:pred-lift-new}
\begin{array}{rl}
[Fn,QA] & \to  Q\wtT A\\
\phi & \mapsto \lambda\circ\wtT(\check\phi) 
\end{array}
\end{equation}
\end{definition}

\begin{example}\label{exle:pred-lift}
  Consider $B=\Pom$, $T=\{*\}+\Act\times\Id$, $\wtT(A)=F\{*\} +
  \Act\cdot A$. As in Example~\ref{exle:B=Pom:2}, we identify
  $F\emptyset$ with $\{\emptyset\}$ and $\wtT^n(F\emptyset)$ with
  $\Pom(1+\Act +\ldots \Act^n)$. The initial and final $\wtT$-algebras
  are then $\Pom(\Act^*)$ and $\P(\Act^*)$,
  respectively. 
  Recall that $QA=[A,F1]=[A,\bbtwo]$ and we write
  $0,1\in\bbtwo$. Further note that, for finite $n$, there is a
  bijection $UQFn=U[Fn,\bbtwo]\cong\Set(n,2)\cong Bn= UFn$ which
  extends to a semi-lattice isomorphism $QFn\cong Fn$.

  In order to obtain the clause for $\tick$, we instantiate
  (\ref{equ:pred-lift-new}) with $n=\emptyset$ (because $\tick$ is a
  constant) and let $\lambda_\tick$ be the unique isormorphism
\begin{equation}
\wtT QF\emptyset\cong \wtT F\emptyset = F\{*\}
  + \Act\cdot F\emptyset \cong F\{*\} \too \bbtwo.
\end{equation}
Consider $A$ and $\phi:F\emptyset\to QA$ and $\check\phi:A\to
QF\emptyset\cong F\emptyset$. This gives us the semantics of $\tick$
as follows.
$\delta_A(\tick)\in Q\wtT A$ as in \eqref{equ:deltarho} is the map
\begin{equation}\label{equ:sem-tick}
  \xymatrix{
    F\{*\} + \Act\cdot A \ar[r]^{\quad\quad\quad\delta(\tick)}
                          \ar[d]_{F\{*\}+\Act\cdot \check\phi}
    & \bbtwo \\
    F\{*\} + \Act\cdot F\emptyset \ar[ur]_{\lambda_\tick}
  }
\end{equation}
Finally, putting this together with (\ref{equ:tr-semantics}) and
(\ref{equ:sem-set-coalg}) we find that, as expected,
$$x\Vdash \tick \ \Leftrightarrow *\in \gamma(x).$$

In order to obtain the clause for $\langle a \rangle\phi$, we
instantiate (\ref{equ:pred-lift-new}) with $n=1$ and let $\lambda_a$
be given by the map
\begin{equation} \wtT QF1 \cong \wtT F1 = F\{*\} +
  \Act\cdot F1 \too \bbtwo
\end{equation}
which sends all generators $*$ and $b\in A, b\not=a$ to $0$ and $a$ to
1. Consider $A$ and choose some $\phi:F1\to QA$. Note that
$\check\phi:A\to QF1\cong F1\cong\twobb$. This gives us the semantics
of $\langle a \rangle\phi$ as follows.
$\delta(\langle a \rangle\phi)\in Q\wtT A$
as in \eqref{equ:deltarho} is the map
\begin{equation}\label{equ:sem-actions}
  \xymatrix{
    F\{*\} + \Act\cdot A \ar[r]^{\quad\quad\quad\delta(\langle a
\rangle\phi)}
                          \ar[d]_{F\{*\}+\Act\cdot \check\phi }
    & \bbtwo \\
    F\{*\} + \Act\cdot F1 \ar[ur]_{\lambda_{a}}
  }
\end{equation}
Finally, putting this together with (\ref{equ:tr-semantics}) and
(\ref{equ:sem-set-coalg}) we find that, as expected,
$$x\Vdash \langle a \rangle\phi  \ \Leftrightarrow \ (a,x')\in\gamma(x)\ \textrm{and}
\ x'\Vdash\phi.$$
\end{example}

Every collection of predicate liftings defines a functor.

\begin{definition}
  Given a collection of predicate liftings $\Lambda$ let $L_\Lambda
  A=F\coprod_{\lambda\in\Lambda} [F(n_\lambda),A]$, where $n_\lambda$
  is the arity of $\lambda$. The semantics $\delta_\Lambda$ acts on a
  generator $(\lambda,\phi)\in Q\wtT QFn\times [Fn,QA]$ as given by
  \eqref{equ:pred-lift-new}.
\end{definition}

\begin{example}
  Let $\Lambda=\{\lambda_\tick\}\cup\{\lambda_a\mid a\in\Act\}$ as in
  Example~\ref{exle:pred-lift}. Then $L_\Lambda A \cong F1 + \Act\cdot
  FUA$ and $\delta_\Lambda$ is given by \eqref{equ:sem-tick} and
  \eqref{equ:sem-actions}.  
\end{example}

It is possible to incorporate logical laws into the functor. 

\begin{example}\label{exle:LambdaE}
  Let $\Lambda=\{\lambda_\tick\}\cup\{\lambda_a\mid a\in\Act\}$ as in
  Example~\ref{exle:pred-lift} and consider the set $E$ of equations
  given by \eqref{equ:exle1-axioms}. Then $L_{\Lambda E}  \cong F1 +
  \Act\cdot\Id$ and $\delta_{\Lambda E}$ is given by
  \eqref{equ:sem-tick} and \eqref{equ:sem-actions}. Furthermore, we
  have
\begin{equation}
\label{equ:deltaLE-kappa}
\vcenter{
\xymatrix@R=5pt{
F1+Act\cdot Q \ar[dd]_\cong\ar[dr]^{\kappa} & \\
& Q\wtT \\
L_{\Lambda E}Q \ar[ur]_{\delta_{\Lambda E}}
}}
\end{equation}
where, on finite $A$, $\kappa_A$ is the isomorphism
\begin{equation}
\label{equ:kappa}
\xymatrix@R=5pt{
  F1+Act\cdot QA \ar[r] 
& F1\times\prod_\Act QA \ar[r] 
& Q(F1+Act\cdot A) \ar[r]^{}
& Q\wtT A 
}
\end{equation}
where the first iso comes from \eqref{equ:co-prod}, the second is due
to $Q$ being a hom-functor, and the third is from the definition of
$\wtT$.
\end{example}

To summarise, we have extracted from the example in
Section~\ref{sec:examples} a general framework that allows to define
trace logics for general functors $T$ and monads $B$ satisfying
Assumption~\ref{ass:fin}.

\subsection{A generic trace logic}\label{sec:gen-logic}

In this section, we show how to define a logic $(L_T,\delta_T)$ for
general functors $T$ and monads $B$ satisfying
Assumption~\ref{ass:fin}. We show that the example from the previous
section arises in that way.

\begin{definition}\label{def:L} 
  The functor $L_{T}:\BAlg\to\BAlg$ is defined on finitely generated
  free algebras $Fn$ as $L_{T}Fn=Q\wtT QFn$. Since every $A\in\BAlg$
  is a colimit of finitely generated free algebras, this extends
  continuously to all $A\in\BAlg$.
\end{definition}
  
\begin{definition}\label{def:LT}
  The semantics $\delta_T : L_{T} Q \to Q \wtT$ is given by
  considering $QA$ as a colimit $\phi_i:Fn_i\to QA$, which is, by
  construction, preserved by $L_T$. More explicitly,
  $(\delta_{T})_{X}$ is the unique arrow making the following diagram  
  \begin{equation}\label{equ:deltaT}
    \xymatrix{
      QA & L_{T} QA \ar[r]^{(\delta_T)_A} & Q\wtT A\\
      Fn_i\ar[u]^{\phi_i} & L_TFn_i \ar[u]^{L_T\phi_i}\ar[r]^{=} & Q\wtT QFn_i
      \ar[u]_{Q\wtT\check \phi_i} }
   \end{equation}
   commute for each $i$; as in \eqref{equ:notation-predicates}, the
   arrow $\check \phi_i$ comes from applying the isomorphism
   $\BAlg(Fn_i, QA)\cong \BAlg(A,QF n_i)$ to $\phi_i$.
\end{definition} 

To show that the example of the previous section is actually the
generic one, we need a lemma helping us to compare the two logics.

\begin{lemma}
  Let $(L,\delta), (L',\delta')$ be two logics and $\rho, \rho'$ as in
  \eqref{equ:deltarho}. If there is an isomorphism $\alpha:LJ\to L'J$
  such that for all finite sets $n$ we have
\begin{equation}
\vcenter{
\xymatrix@R=5pt{
LFn \ar[dd]_{\alpha_n}\ar[dr]^{\rho} & \\
& Q\wtT QFn \\
L'Fn \ar[ur]_{\rho'}
}}
\end{equation}
then this extends to an isomorphism $\beta:L\to L'$ of logics, ie,
$\beta$ satisfies
\begin{equation}
\vcenter{
\xymatrix@R=5pt{
LQ \ar[dd]_{\beta Q}\ar[dr]^{\delta} & \\
& Q\wtT  \\
L'Q \ar[ur]_{\delta'}
}}
\end{equation}
Moreover, $\beta_{Fn}=\alpha_n$.
\end{lemma}

Consequently, any collection of isomorphisms $Ln\to Q\wtT QFn$,
$n\in\mathbb{N}$,  defines the same logic, or, more precisely:

\begin{corollary}\label{cor:LT}
  The generic logic $L_T$ is determined up to isomorphism, that is,
  for any other logic $(L,\delta)$ with the $LFn\to Q\wtT QFn$ as in
  \eqref{equ:deltarho} being isos, there is a unique isomorphism $L\to L_T$
  such that 
\begin{equation}
  \vcenter{
    \xymatrix@R=5pt{
      LQ \ar[dd]_{}\ar[dr]^{\delta} & \\
      & Q\wtT  \\
      L_TQ \ar[ur]_{\delta_T}
    }}
\end{equation}
\end{corollary}

Finally, we can show that the generic logic of this subsection agrees
with the logic defined, in different ways, by
\eqref{equ:exle1-syntax}-\eqref{equ:exle1-axioms}, or again in
Example~\ref{exle:two} or in Example~\ref{exle:LambdaE}.

\begin{proposition}
  Going back to Example~\ref{exle:LambdaE}, there is an isomorphism
  such that
\begin{equation}
\xymatrix@R=5pt{
L_{\Lambda E}Q\ar[dd]_\cong\ar[dr]^{\delta_{\Lambda E}} & \\
& Q\wtT \\
L_TQ \ar[ur]_{\delta_T}
}
\end{equation}
\end{proposition}

\begin{proof}
  We write $(L,\delta)$ for $(L_{\Lambda E},\delta_{\Lambda E})$ and
  $\rho$ for the natural transformation as in
  \eqref{equ:deltarho}. According to Corollary~\ref{cor:LT}, it is
  enough to show that $\rho_{Fn}:LFn\to Q\wtT QFn$ is an
  isomorphism. From the proof of Proposition~\ref{prop:pred-lift}, we
  know that $\rho_{Fn}=\delta_{QFn}\circ L\eta$. Since $\eta$ is an
  isomorphism for finite semi-lattices, the result now follows from
  $\delta_{QFn}$ being iso, see Example~\ref{exle:LambdaE}.
\end{proof}

Finally, Definition~\ref{def:LT} does not depend on the choice of a
partiuclar $T$ or $B$, so we can summarise this section as follows.

\begin{theorem}
  For every monad $B$ on $\Set$ and functor $T:\Set\to\Set$ satisfying
  Assumption~\ref{ass:fin} there is a generic trace logic.
\end{theorem}

Of course, given $B$ and $T$, the real work consists in finding a good
explicit description of the generic logic. We have illustrated this
for the moment only with one example. 

We can apply the general framework to obtain results about generic
logics. For example, we have

\begin{theorem}
  The logic of Example~\ref{exle:LambdaE} is expressive and complete.
\end{theorem}

\begin{proof}
  We write $(L,\delta)$ for $(L_{\Lambda E},\delta_{\Lambda E})$. The
  proof is straightforward due to the following facts: $B$ and $\wtT$
  preserve finite algebras and on finite algebras we have that
  $\delta$ is an isomorphism. In detail:

  Expressiveness means that any two non-trace equivalent states can be
  separated by a formula. Consider a coalgebra $X\to BTX$ with
  $x,x'\in X$ and suppose $x$ accepts trace $t$ and $x'$ does
  not. Since the initial $L$-algebra is the free $B$-algebra over the
  set of traces, $t$ can be considered as a formula and we have
  $x\Vdash t$ and $x'\not\Vdash t$.

  Completeness means that if $L$ does not prove $\phi=\phi'$, then
  there must be a coalgebra $X\to BTX$ and $x\in X$ such that, wlog,
  $x\Vdash\phi$ and $x\not\Vdash\phi'$. Since $\phi$ and $\phi'$ appear
  at some stage $n$ in the initial algebra construction of $L$, the
  semantics of $\phi$ and $\phi'$ is determined at stage $n$. Since
  $\delta$ is an iso on finite algebras, the images
  of $\phi$ and $\phi'$ in $Q\wtT^nF\emptyset$ are different. It
  follows from a standard argument that there is a $\wtT$-coalgebra
  $\wtgamma:\wtT^nF\emptyset\to\wtT(\wtT^nF\emptyset)$ that refutes
  the equation $\phi=\phi'$. In particular,
  $\sem{\phi}_\wtgamma\not=\sem{\phi'}_\wtgamma$ are two different
  morphisms $FT^n\emptyset=Q\wtT^nF\emptyset\to\twobb$, so they must
  differ on some generator $\eta_X(x)$ where $\eta_X:X\to BX$ maps
  elements $x$ to singletons $\{x\}$. It follows now from
  Proposition~\ref{prop:wtgamma} that the $(B,T)$-coalgebra
  $U\gamma\circ\eta_X:X\to BTX$ contains a state $x$ with
  $x\Vdash\phi$ and $x\not\Vdash\phi$.
\end{proof}



\end{document}